\documentclass[letterpaper, 10 pt, conference]{ieeeconf}

\IEEEoverridecommandlockouts                              
\usepackage{bbm}
\usepackage[english]{babel}
\usepackage{amsmath,amsfonts,amssymb, amsthm}
\usepackage{graphicx}
\usepackage[colorlinks=true, allcolors=blue]{hyperref}
\usepackage{parskip}
\usepackage{tabularx}
\usepackage{bm}
\usepackage{hhline}
\usepackage{multirow}
\usepackage{algorithm}
\usepackage{algpseudocode}
\usepackage{etoolbox}
\usepackage{subcaption}
\usepackage{algpseudocode}
\usepackage[dvipsnames]{xcolor}
\usepackage{mathtools}
\usepackage{nth}
\usepackage{diagbox}
\usepackage{float}
\usepackage{tabstackengine}
\usepackage{tablefootnote}
\usepackage{dsfont}
\usepackage{wrapfig}

\DeclareMathOperator*{\argmin}{arg\,min}

\patchcmd{\thebibliography} 
  {\settowidth}
  {\setlength{\parsep}{0pt}\setlength{\itemsep}{0pt plus 0.1pt}\settowidth}
  {}{}

\newtheorem{theorem}{Theorem}
\newtheorem{proposition}{Proposition}
\newtheorem{corollary}{Corollary}

\renewcommand{\eqref}[1]{(\ref{#1})}
\renewcommand{\t}{^{\mbox{\tiny\sf T}}} 

\newcommand{\N}{\mathcal{N}}
\newcommand{\KL}{\mathrm{D}_{\mathrm{KL}}}

\newcommand{\D}{\mathcal{D}}

\newcommand{\R}{\mathbb{R}}
\newcommand{\E}{\mathbb{E}}
\newcommand{\tr}{\mathrm{tr}}
\renewcommand{\d}{\mathrm{d}}

\allowdisplaybreaks[3] 

\title{\LARGE \bf{Schr\"odinger Bridges and Density Steering Problems for \\ Gaussian Mixtures Models in Discrete-Time}}
\author{George Rapakoulias, Fengjiao Liu, and Panagiotis Tsiotras}

\begin{document}

\maketitle

\begin{abstract}
    In this work, we revisit the discrete-time Schr\"{o}dinger Bridge (SB) and Density Steering (DS) problems for Gaussian mixture model (GMM) boundary distributions. 
    Building on the existing literature, we construct a set of feasible Markovian policies that transport the initial distribution to the final distribution, and are expressed as mixtures of elementary component-to-component optimal policies.
    We then study the policy optimization within this feasible set in the context of discrete-time SBs and density-steering problems, respectively.
    We show that for minimum-effort density-steering problems, the proposed policy achieves the same control cost as existing approaches in the literature. 
    For discrete-time SB problems, the proposed policy yields a cost smaller than or equal to that in the literature, resulting in a less conservative approximation.
    Finally, we study the continuous-time limit of our proposed discrete-time approach and show that it agrees with recently proposed approximations to the continuous-time SB for GMM boundary distributions.
    We illustrate this new result through two numerical examples.
\end{abstract}

\section{Introduction}

Controlling the density of dynamical systems is a promising and rapidly evolving area of control theory~\cite{chen2021controlling}. 
The goal is to design stochastic processes that steer the probability distribution of a dynamical system from a prescribed initial distribution to a desired terminal one within a finite time horizon. 
Among the most prominent formulations of this problem are Schr\"odinger Bridge problems, which seek the process closest to a given reference dynamics while satisfying marginal distribution constraints. 
SBs have attracted significant attention due to their applicability to generative modeling and machine learning~\cite{chen2022likelihood, shi2023diffusion, rapakoulias2025go}, mean-field and multi-agent control~\cite{chen2018steering, liu2022deep, rapakoulias2025steering}, and autonomy and robotics applications~\cite{pilipovsky2024computationally, saravanos2023distributed, rapakoulias2023discrete, kumagai2025hands, mei2025time}.

Most existing approaches in the generative modeling and mean-field control literature focus on continuous-time formulations. 
These methods are well-suited for high-dimensional distributions that are only available through samples and have led to scalable learning-based algorithms. 
In contrast, many autonomy and control applications naturally operate in discrete-time, allowing them to leverage finite-dimensional convex optimization techniques, such as semidefinite programming~\cite{bakolas2018finite}. 
In this setting, existing approaches often focus on Gaussian distributions and arise in the context of explicitly controlling the statistics of stochastic dynamical systems, such as covariance steering problems. 
Recently, methods capable of handling richer distribution classes, including Gaussian mixture models, have been proposed both in discrete-time~\cite{balci2023density, kumagai2024chance} and in continuous time~\cite{rapakoulias2025go, rapakoulias2025steering, mei2024flow}.

In this work, we argue that a deeper understanding of discrete-time methods capable of handling general, high-dimensional distributions can lead to algorithms with practical advantages in both the control and the generative AI domains.
In generative modeling, discrete-time formulations enable faster inference than continuous-time diffusion-based methods~\cite{gushchin2024adversarial}, by only requiring the traversal of a Markov chain with a small number of time steps, instead of approximately integrating continuous-time stochastic differential equations (SDEs).
In autonomy applications, discrete approaches align more naturally with digital control implementations.
Realizing these benefits, however, requires algorithms capable of handling expressive classes of distributions beyond Gaussian distributions.

In this work, we revisit the discrete-time Gaussian mixture steering algorithm originally proposed by~\cite{balci2023density}, and introduce a modified formulation that bridges discrete-time density steering and continuous-time SB methods.
Our modification yields a Markovian control policy that depends only on the current state, whereas the policy in~\cite{balci2023density} requires memory of a latent random variable sampled at the beginning of the horizon. 
We show that, for minimum-effort density steering problems, the proposed policy achieves the same control cost as~\cite{balci2023density}, while for discrete-time SB problems, the proposed policy gives a cost that is smaller than or equal to that of~\cite{balci2023density}, yielding a less conservative solution.
Moreover, we demonstrate that, in the continuous-time limit, the proposed formulation recovers the policy structure introduced in~\cite{rapakoulias2025go}.
\section{Preliminaries}

\subsection{Discrete-Time Schr\"odinger Bridge}

The discrete-time Schr\"odinger Bridge (dtSB)~\cite{beghi1996relative} is a discrete-time formulation of the well-known Schr\"odinger Bridge problem~\cite{chen2021stochastic, leonard2014survey}.
From a stochastic process perspective, it is an optimization problem in the space of probability measures over discrete-time processes. Formally, given a reference $N$-step Markov chain $q$, the dtSB is the solution of
\begin{equation} \label{dtSB}
    \min_{p \in \D(\rho_0, \rho_N)}  \KL(p \| q) 
\end{equation}
where $\D(\rho_0, \rho_N)$ is a set of processes with fixed marginals $\rho_0, \rho_N$ at time steps $k=0, N$, and $q$ is usually associated with the discrete-time Gaussian random walk $x_{k+1} = x_k +  \sqrt{\epsilon} \, w_k, \,\, w_k \sim \N(0, I)$.

Solving for the optimal process 
\begin{equation*}
   p^* = \argmin_{p \in \D(\rho_0, \rho_N)}  \KL(p \| q) 
\end{equation*}
in \eqref{dtSB} for arbitrary marginal distributions $\rho_0, \rho_N$, is difficult, in general, and therefore computational methods are usually sought.
A state-of-the-art approach for high-dimensional problems is given in ~\cite{gushchin2024adversarial}, which models the transitions of $p$ via conditional Generative Adversarial Networks (GANs) trained using stochastic optimization schemes.

\subsection{Discrete-Time Linear Density Steering}

For many control problems, rather than finding the optimal process that is as close as possible to a reference process, it is useful to study discrete-time minimum control effort processes subject to marginal distribution constraints.
Specifically, instead of optimizing over arbitrary Markov chains, as in \eqref{dtSB}, one solves
\begin{subequations} \label{dtDS}
    \begin{align}
        \min_{u_k}    \quad & \sum_{k=1}^{N-1} \E \left[ \| u_k \|^2 \right], \\
        \textrm{s.t.} \quad & x_{k+1} = A_k x_k + B_k u_k + D_k w_k, \label{dtDS:LTI}\\
        & x_0 \sim \rho_0, \, \quad x_N \sim \rho_N,
    \end{align}
\end{subequations}
where $x_k \in \R^n$, $A_k \in \R^{n \times n}$, $B_k \in \R^{n \times m}$, $D_k \in \R^{n \times \ell}$.
The formulation \eqref{dtDS} has been extensively studied, in linear density control problems such as the covariance steering framework~\cite{chen2015optimal, bakolas2018finite, balci2022exact, rapakoulias2023discrete}, its generalization to problems with Gaussian mixture distributions~\cite{balci2023density, kumagai2024chance}, or in a more abstract setting through the lens of optimal transport \cite{terpin2024dynamic}.

Unlike their continuous-time counterparts, Problems \eqref{dtSB} and \eqref{dtDS} are not equivalent in general, because even in the case where the linear time-invariant prior \eqref{dtDS:LTI} matches the random walk transitions $x_{k+1} = x_k + \sqrt{\epsilon} \, w_k$, the optimal solution of \eqref{dtSB} does not share the same noise intensities as the reference process, i.e., it is not associated with a state space model of the form $x_{k+1} = x_k + u_k(x_k) + \sqrt{\epsilon} \,  w_k$~\cite{beghi1996relative}.

\subsection{Gaussian Discrete-Time Schr\"odinger Bridge}

In the case when $\rho_0 = \N(\mu_0, \Sigma_0), \rho_N = \N(\mu_N, \Sigma_N)$ are Gaussian distributions with the specified means and covariances, and for a reference process $x_{k+1} = x_k + \sqrt{\epsilon} \, w_k$, \eqref{dtSB} can be solved analytically.
Specifically, we can derive the closed-form transition kernels and the marginal distributions of the process $p$ using~\cite{ito2023maximum, liu2022optimal, levy1994discrete, lambert2025lqr}.
\begin{proposition} \label{prp:Gaussian-dtSB}
    Let $\rho_0 = \N(\mu_0, \Sigma_0)$, $\rho_N = \N(\mu_N, \Sigma_N)$ and consider a reference process $q$ associated with $x_{k+1} = x_k + \sqrt{\epsilon} \, w_k$.
    Then, the transition kernels and the marginal distributions of the optimal process $p^*$ that solves the dtSB problem \eqref{dtSB} are given by 
\begin{align}
& p^*_{k+1|k}(x_{k+1} |x_k) = \N(\mu_{k+1|k}, \Sigma_{k+1|k}), \\
& p^*_k = \N(\mu_k, \Sigma_k),
\end{align}
where, 
\begin{align}
    & \mu_k=\left(1-\frac{k}{N}\right)\mu_0+\frac{k}{N}\mu_N, \\
    & \Sigma_k=\big(P_k^{-1}+Q_k^{-1}\big)^{-1}, \\
    & \mu_{k+1|k} = \mu_{k+1}+\big(I_n - \epsilon Q_k^{-1}\big)(x_k-\mu_k), \\
    & \Sigma_{k+1|k} = \epsilon \, Q_{k+1}Q_k^{-1} 
    = \epsilon \big(I_n - \epsilon Q_k^{-1}\big),
\end{align}
and
\begin{align}
& P_k=P_0+k\epsilon I_n ,\qquad Q_k=Q_0-k\epsilon I_n , \label{Qk} \\
& Q_0 \! = \! \epsilon N
\Sigma_0^{\frac{1}{2}}
\Bigg( \!\Sigma_0 \!+\!\frac{\epsilon N}{2}I_n \!- \!\!\Big(\Sigma_0^{\frac{1}{2}}\Sigma_N\Sigma_0^{\frac{1}{2}}+\frac{\epsilon^2 N^2}{4}I_n\Big)^{\frac{1}{2}}
\Bigg)^{-1} \!\!\!\!\! \Sigma_0^{\frac{1}{2}}, \label{Q0} \\
& P_0=\big(\Sigma_0^{-1}-Q_0^{-1}\big)^{-1}. \nonumber
\end{align}
Furthermore, the optimal cost is given by 
\begin{multline*}
    J^{\mathrm{GSB}} = \frac{1}{2} \Big[
    2 \, \tr\big(\Sigma_0 Q_0^{-1}\big) 
    - \ln (\det V)
    + n \ln (\epsilon N)
    - n 
    \\
    + \frac{\|\mu_N - \mu_0\|^2}{\epsilon N}
    + \frac{\tr(\Sigma_N - \Sigma_0)}{\epsilon N}
    \Big],
\end{multline*}
where $V = \Sigma_N - \big(I_n - \epsilon N Q_0^{-1}\big) \Sigma_0 \big(I_n - \epsilon N Q_0^{-1}\big)$. 
\end{proposition}
\begin{proof}
The proof follows from the closed-form solution provided in~\cite[Theorem 3]{ito2023maximum}. 
\end{proof}

\subsection{Gaussian Discrete-Time Linear Density Steering} Unlike \eqref{dtSB}, even in the case where $\rho_0, \rho_N $ are Gaussian, Problem \eqref{dtDS} cannot be solved analytically. 
Its solution, however, can be characterized by a pair of discrete-time coupled Riccati equations, which can be solved numerically~\cite{liu2022optimal}, or through a convex semidefinite program~\cite{rapakoulias2023discrete}.
Since both methods are well-established in the literature, we omit including them here due to space limitations. 


\section{Schr\"odinger Bridge for Mixture Distributions}

Consider two mixture distributions
\begin{equation} \label{finite_mix}
\rho_0(x_0) = \sum_{i=1}^{N_1} \alpha^i \mu^i(x_0), \quad \rho_N(x_N) = \sum_{j=1}^{N_2} \beta^j \nu^j(x_N),
\end{equation}
where $\{\mu^i, \alpha^i\}_{i=1}^{N_1}, \{\nu^j, \beta^j\}_{j=1}^{N_2}$ are two sets of component-weight pairs, respectively.

The goal of this section is to approximate the solution of the discrete-time SB problem \eqref{dtSB} for the mixture distributions \eqref{finite_mix} and a reference process $q$, i.e., approximate the solution of
\eqref{dtSB}
%
given access to a family of Markov chains $p^{ij}$, each solving a component $i$ to component $j$ discrete-time SB problem. 
Specifically, let the family of $p^{ij}$ be Markov chains with $N$ total time steps, parameterized by the indices $(i,j) \in \{1, \dots, N_1\} \times \{1, \dots, N_2\}$. Denote their joint probabilities by 
\begin{equation}
    p^{ij}(x_0, \dots, x_N) = p^{ij}_0(x_0) \prod_{k=0}^{N-1} p^{ij}_{k+1|k} (x_{k+1}|x_k),
\end{equation}
and assume that $p_0^{ij} = \mu^{i}, \; p_N^{ij} = \nu^j$, for all $i,j$, i.e., each Markov chain $p^{ij}$ transports samples from $\mu^i$ to $\nu^j$. 
So far, $p^{ij}$ does not necessarily have to be the optimal component-to-component SB solution, but rather a feasible bridge between $\mu^i$ and $\nu^j$.  
As we will see later, choosing $p^{ij}$ as the component-to-component SB is optimal within a certain class of policies.

One way of constructing a Markov chain $p$ that transports samples from the mixture $\rho_0$ to the mixture $\rho_N$, as defined in equation \eqref{finite_mix}, given the individual family of Markov chains $p^{ij}$ transporting $\mu^i$ to $\nu^j$, is through the following theorem. 
\vspace{2mm}
\begin{theorem} \label{thm:mixing}
Let $p^{ij}(x_0, \dots, x_N)$, where $i \in \{1, \dots, N_1\}$ and $j \in \{1, \dots, N_2\}$, be a family of $N$-step Markov chains with $p_0^{ij} = \mu^{i}, \; p_N^{ij} = \nu^j$, for all $i,j$, and define the scalar mixing weights $\lambda_{ij} \geq 0$ such that $\sum_{i} \lambda_{ij} = \beta^j$ and $\sum_{j} \lambda_{ij} = \alpha^i$.
Then, starting from $x_0 \sim \rho_0$, and following the transition probabilities
\begin{equation}
    p_{k+1|k}(x_{k+1}|x_k) = \sum_{ij} p^{ij}_{k+1|k} (x_{k+1} |x_k) \frac{\lambda_{ij} p^{ij}_k(x_k)}{p_k(x_k)}, \label{mtrans}
\end{equation}
where, 
\begin{equation}
    p_k(x_k) = \sum_{ij} \lambda_{ij} p^{ij}_{k} (x_k), \quad
    k = 0, \dots, N-1,
    \label{mflow}
\end{equation}
the resulting distributions for $x_k$ will be given by equation \eqref{mflow}. 
\end{theorem}
%
\begin{proof}
    It is a simple consequence of marginalization, that $p_{k+1}(x_{k+1}) = \int p_{k+1|k}(x_{k+1}|x_k) p_k(x_k) \, \d x_k$. Using this, it is easy to show that if $p_{k} = \sum_{ij}\lambda_{ij} p^{ij}_k$ holds for some time step $k$, then $p_{k+1} = \sum_{ij}\lambda_{ij} p^{ij}_{k+1}$. The result follows by recursively applying this formula, and observing that $p_0 = \rho_0$ due to the constraints on $\lambda_{ij}$. 
\end{proof}
Given the family of Markov processes and transition kernels in equations \eqref{mflow}, \eqref{mtrans}, parameterized by $\lambda_{ij}$ and $p^{ij}$, we are interested in the one that induces the smallest cost in the sense of Problem \eqref{dtSB}.
Substituting the transition probabilities in \eqref{dtSB} and carrying out some calculations, we arrive at the following theorem. 
\begin{theorem} \label{thm:optimality}
    Consider the family of Markov processes $p$, parameterized by $p^{ij}$ and $\lambda_{ij}$ satisfying the constraints of Theorem \ref{thm:mixing}, and with transition probabilities given by \eqref{mtrans}.
    Let $J_{ij}^{\mathrm{SB}}$ be the optimal SB cost of solving \eqref{dtSB} from the $i$-th component of the initial mixture to the $j$-th component of the terminal mixture for a reference process $q$. 
    Then, the KL divergence between $p$ and the reference process $q$ is upper bounded by
    \begin{equation} \label{KL_UB}
        \KL(p \|q) \leq \sum_{ij} \lambda_{ij} J_{ij}^{\mathrm{SB}}.
    \end{equation}
\end{theorem}
\vspace{-5mm}
\begin{proof}
We first consider the following upper bound
    \begin{align}
         \KL(p \|q)  &= \KL(p_0 \| q_0) + \sum_k \E_{x_k} \left[ \KL(p_{k+1|k} \| q_{k+1|k}) \right] \nonumber \\
         &  \leq \sum_k \E_{x_k} \left[ \sum_{i,j} \frac{\lambda_{ij} p_k^{ij}}{p_k} \KL(p^{ij}_{k+1|k} \| q_{k+1|k}) \right]  \label{pf2:second_line} \\
        &  = \sum_{k,i,j} \int \lambda_{ij} p_k^{ij} \KL(p^{ij}_{k+1|k} \| q_{k+1|k}) \, \d x_k \nonumber \\
        & = \sum_{i,j} \lambda_{ij} \sum_{k} \E_{p_k^{ij}} \left[ \KL(p^{ij}_{k+1|k} \| q_{k+1|k}) \right] \nonumber \\
        & = \sum_{i,j} \lambda_{ij} \KL(p^{ij} \| q), \label{pf2:last_line}
    \end{align}
    where the inequality in \eqref{pf2:second_line} follows from the convexity of the KL divergence~\cite[Theorem 2.7.2]{cover2012elements}. 
    Since $\lambda_{ij} \geq 0$  in \eqref{pf2:last_line}, the upper bound is minimized by selecting $p^{ij}$ to be the process solving the ($i$-$j$)-SB, yielding the desired result.
\end{proof}
\vspace{-3mm}
Using the upper bound of Theorem \ref{thm:optimality} as a proxy for the loss function \eqref{dtSB}, we can optimize for $\lambda_{ij}$ to obtain an approximation of the discrete-time Schr\"odinger Bridge by solving 
\begin{equation*}
    \min_{\lambda_{ij} \geq0} \sum_{ij} \lambda_{ij} J^{\mathrm{SB}}_{ij} \, ; \,\, \mathrm{s.t.} \,\, \Big\{ \sum_{i} \lambda_{ij} = \beta^j,  \sum_{j} \lambda_{ij} = \alpha^i \Big\}. 
\end{equation*}
\section{Discrete-Time Density Steering}
The goal of this section is to solve problems of the form 
\begin{subequations} \label{GMM_dtDS}
    \begin{align} 
        \min_{u_k}    \quad & J= \E \Big[\sum_{k=1}^{N-1} \| u_k \|^2 \Big], \label{GMM_dtDS:cost}\\
        \textrm{s.t.} \quad & x_{k+1} = A_k x_k + B_k u_k + D_k w_k, \label{GMM_dtDS:dyn} \\
        & x_0 \sim \rho_0, \quad x_N \sim \rho_N, \label{GMM_dtDS:BC}
    \end{align}
\end{subequations}
where $\rho_0, \rho_N$ are mixtures of the form \eqref{finite_mix}, given a family of component-wise solutions. 
To this end, let $u_{k|ij}$ be a policy solving \eqref{GMM_dtDS:cost}, \eqref{GMM_dtDS:dyn} with boundary distributions the $i$-th component of the initial mixture and the $j$-th component of the terminal mixture. 
Denoting the Markov chain associated with the state space model \eqref{GMM_dtDS:dyn} for $u=u_{t|ij}$ with $p^{ij}$, the corresponding $k$-th marginal and transition kernels by $p^{ij}_k, p^{ij}_{k+1|k}$, respectively, and using Theorem \ref{thm:mixing}, the transition function \eqref{mtrans} will transport the GMM $\rho_0$ to $\rho_N$.
Using Bayes rule, we translate the statement about mixing the transition function in \eqref{mtrans}, to mixing the control policies $u_{t|ij}$, in the following corollary.

\begin{corollary}
        Starting from $x_0 \sim \rho_0$, and following, for each $k = 0, \dots, N-1$, 
        the random policy
    \begin{equation} \label{random_pol}
        u_k = u_k^{ij}(x_k) \quad \mathrm{w.p.} \quad \frac{\lambda_{ij} p^{ij}_k(x_k)}{p_k(x_k)},
    \end{equation}
    the resulting distributions for $x_k$ will be given by equation \eqref{mflow}.
\end{corollary}
\begin{proof}
It suffices to show that using policy \eqref{random_pol} on system \eqref{GMM_dtDS:dyn} results in the state transition probability given by \eqref{mtrans}.

By Bayes rule, we know that 
\begin{equation} \label{bayes_rule_for_u}
    p_{k+1|k}(x_{k+1}|x_k)=\int p(x_{k+1}|x_k, u_k) p(u_k|x_k) \, \d u_k. 
\end{equation}
By substituting the control distribution \eqref{random_pol} in the RHS of \eqref{bayes_rule_for_u}, we obtain
\begin{align}
    & \int p(x_{k+1}|x_k, u_k) p(u_k|x_k) \, \d u_k \nonumber \\
    & = \int p(x_{k+1}|x_k, u_k) \sum_{i,j} \delta\big(u_k - u^{ij}_k(x_k) \big) \frac{\lambda_{ij} p^{ij}_k(x_k)}{p_k(x_k)} \, \d u_k \nonumber \\
    & = \sum_{i,j} \int p(x_{k+1}|x_k, u_k)  \delta\big(u_k - u^{ij}_k(x_k) \big) \, \d u_k \frac{\lambda_{ij} p^{ij}_k(x_k)}{p_k(x_k)}\nonumber \\
    & = \sum_{i,j} p^{ij}(x_{k+1}|x_k) \frac{\lambda_{ij} p^{ij}_k(x_k)}{p_k(x_k)}. \label{rhs}
\end{align}
Equation \eqref{rhs} matches \eqref{mtrans}, which guarantees that the state distribution will evolve according to \eqref{mflow}. This concludes the proof. 
\end{proof}

Regarding the optimization over $\lambda_{ij}$ for density steering problems, we provide the following result: 
\begin{theorem} \label{optimality_DS}
Let $J_{ij}^{\mathrm{DS}}$ be the cost of the ($i$-$j$)-th conditional policy $u_{t|ij}$, i.e., the cost \eqref{GMM_dtDS:cost} from the $i$-th component of the initial to the $j$-th component of the terminal mixtures \eqref{GMM_dtDS:BC}.
The control cost \eqref{GMM_dtDS:cost} for the density steering problem \eqref{GMM_dtDS} is given by 
\begin{equation} \label{cost_DS}
          J = \sum_{ij} \lambda_{ij} J_{ij}^{\mathrm{DS}}.
\end{equation}
\end{theorem}
\begin{proof}
We can write the cost \eqref{GMM_dtDS:cost} as
\begin{align*}
    J & = \sum_k \E_{p_k(x_k)
    } \left[\E_{p(u_k|x_k)} \left[  \|u_k \|^2 \right] \right] \\
      & = \sum_{i,j,k} \! \int \!\! p_k(x_k) \frac{\lambda_{ij} p^{ij}_k(x_k)}{p_k(x_k)} \!\!\! \int \!\! \delta \big(u_k \! - \! u_k^{ij}(x_k) \big) \| u_k \|^2  \d u_k  \d x_k \\
      & =  \sum_{i,j,k}   \int   \lambda_{ij} p^{ij}_k(x_k)  \| u^{ij}_k(x_k) \|^2 \, \d x_k \\
      & = \sum_{ij} \lambda_{ij} J_{ij}^{\mathrm{DS}},
\end{align*}
which completes the proof of \eqref{cost_DS}. 
\end{proof}
In contrast to the discrete-time SB problem, which is optimized approximately using the upper bound of and Theorem~\ref{thm:optimality}, Theorem~\ref{optimality_DS} allows for exact calculation of the optimal policy of the discrete-time density steering problems, within the family parameterized by $\lambda_{ij}$, by solving 
\begin{equation*}
    \min_{\lambda_{ij} \geq0} \sum_{ij} \lambda_{ij} J^{\mathrm{DS}}_{ij} \, ; \,\, \mathrm{s.t.} \,\, \Big\{ \sum_{i} \lambda_{ij} = \beta^j,  \sum_{j} \lambda_{ij} = \alpha^i \Big\}. 
\end{equation*}

\section{Connections with Existing Literature}
\subsection{Randomizing Once vs Randomizing at Every Step}
A similar ``randomized-mixing'' approach for constructing processes that steer between GMMs has been proposed in~\cite{balci2023density, balci2025constrained}, and was further explored in the context of chance constraints~\cite{kumagai2024chance} and distributionally robust control~\cite{li2026distributionally}. 
The difference between the approach in~\cite{balci2023density} and the proposed approach is that the randomization in~\cite{balci2023density} occurs once, at the beginning of the control horizon, as opposed to policy \eqref{mtrans}, where the randomization occurs at every time step.
Although both methods provide feasible solutions to either the density steering or the SB problem, randomizing at every time step yields a Markovian policy, whereas randomizing once at the beginning yields a process with memory.

Except for different properties of the resulting processes due to the Markovian property, it is easy to see that both processes share the same marginal distributions for every time step $k$, and result in the same control cost for the density steering problems (see Theorem \ref{optimality_DS} as opposed to~\cite[Theorem 1]{balci2023density}). 
In the context of SBs, the Markov process parameterized by the transitions \eqref{mtrans} yields a cost that is smaller than or equal to the cost of the policy in~\cite{balci2023density}.
Below, we prove this result:

In the context of Theorem \ref{thm:mixing}, the joint density of the process that corresponds to a single randomization in the first time step, denoted with $r$, calculates to $r = \sum_{ij} \lambda_{ij} p^{ij}$, i.e., $r$ is a mixture of path measures, rather than a mixture of transition functions.
Let $p$ be the path measure of the Markov process with transitions \eqref{mtrans}.
It is trivial to see that $p_k=r_k$, i.e., the marginals at time step $k$ match for both approaches.
Using the definition of $r$ and marginalization, it is also easy to show that $r_{k+1, k}(x_{k+1}, x_k) = p_{k+1, k} (x_{k+1}, x_k)$, i.e., the joint densities for two consecutive time steps match for the two approaches. 
Note that the KL-divergence between two distributions $p$ and $q$ equals 
\begin{equation} \label{KL_decomp}
    \KL(p \| q)  = - \E_p [ \log q]  - h(p),
\end{equation}
where $h(p) = - \int \log p \, \d p $ denotes the differential entropy of $p$.
It is easy to show that $h(p) \geq h(r)$, i.e., randomizing at every time step results in a process with larger entropy than randomizing once:
\begin{align}
    & h(r) =  h(r_0) + \sum_k \E_{x_{0:k} \sim r_{0:k}} \big[h\big(r_{k+1|0:k}(x_{k+1}|x_{0:k}) \big) \big] \nonumber \\
    & \leq  h(r_0) + \sum_k \E_{x_{k} \sim r_k} \big[h \big(r_{k+1|k}(x_{k+1}|x_{k}) \big) \big] \label{cond_entr} \\
    & = h(r_0) + \sum_k \E_{x_{k} \sim r_k} \left[h \left(\frac{r_{k+1,k}(x_{k+1}, x_k)}{r_k(x_{k})} \right) \right] \nonumber \\
    & =  h(p_0) + \sum_k \E_{x_{k} \sim p_k} \big[h\big(p_{k+1|k}(x_{k+1}|x_{k}) \big) \big] \label{two_marginal}\\
    & = h(p) \nonumber
\end{align}
where the notation $x_{0:k}$ denotes the set of random variables $x_0, \dots, x_k$, the inequality in \eqref{cond_entr} holds because conditioning reduces entropy~\cite[Theorem 2.6.5]{cover2012elements}, and \eqref{two_marginal} holds because $r_{k+1, k}(x_{k+1}, x_k) = p_{k+1, k} (x_{k+1}, x_k)$ and $r_k(x_k) = p_k(x_k)$. 

Furthermore, for the maximum likelihood term in \eqref{KL_decomp}, it holds that $\E_p [\log q] = \E_r[\log q]$, since
\begin{align*}
    \E_p [ \log q]  & = \E_{p_0}[\log q_0] + \! \sum_k  \E_{p_{k+1, k}} \Big[ \log q_{k+1|k}(x_{k+1}| x_k) \Big] \\
    & = \E_{r_0}[\log q_0]  + \sum_k  \E_{r_{k+1, k}} \Big[ \log q_{k+1|k}(x_{k+1}| x_k) \Big] \\
    & = \E_r \big[ \log q \big].
 \end{align*}
We therefore deduce that
\begin{equation}
    \KL(p \| q) \leq \KL( r \| q),
\end{equation}
and, therefore, randomizing at every time step has a lower cost in terms of the SB cost function. 

Finally, we note that the process $p$ resulting from the mixture kernel \eqref{mtrans} is the Markovian projection~\cite{gushchin2024adversarial} of the mixture process $r = \sum_{ij} \lambda_{ij} p^{ij}$, that is, $p$ is the solution of  
\begin{equation}
    p = \argmin_{m \in \mathcal{M} } \KL( r \| m), 
\end{equation}
where $\mathcal{M}$ is the set of all Markov processes supported in $\R^{(N+1) \times n}$. 
Due to space limitations, we defer the proof of this result and refer the reader to~\cite{gushchin2024adversarial} for a concise treatment of the Markovian projection for discrete-time processes.

\subsection{Continuous-Time Limit}
In this section, we show that in the limit of small time steps, the discrete-time policy process \eqref{mtrans} converges to the continuous time process of~\cite{rapakoulias2025go}.
To study the convergence of the discrete-time process, we will consider a reference process parameterized by $\Delta T>0$, given by 
\begin{equation} \label{dt_process}
    x_{k+1} = x_k + \Delta T u_t + \sqrt{\epsilon \Delta T} \, w_k, \quad w_k \sim \N(0, I).
\end{equation}
In this setting, consider the setup of Problem \eqref{dtSB}, and study the generator of the process $p$, as $\Delta T \rightarrow 0$. 
Since the state mean and state covariance can be separated, we assume $\mu_0^i = \mu_N^j = 0$ for the moment. 
Using~\cite[Theorem 3]{ito2023maximum}, we can show that the discrete-time optimal process $p^{ij}$ for the $ij$-th component satisfies
\begin{equation*} 
    x_{k+1}^{ij} \hspace{-0.8mm} = \hspace{-0.7mm} \big(\hspace{-0.4mm} I - \epsilon \Delta T (Q_k^{ij})^{-1} \hspace{-0.4mm}\big) x_k^{ij} + \sqrt{\epsilon \Delta T} \big(\hspace{-0.4mm} I - \epsilon \Delta T (Q_k^{ij})^{-1}\hspace{-0.4mm}\big)^\frac{1}{2} w_k^{ij}\!,
\end{equation*}
where $w_k^{ij} \sim \N(0, I)$ and $Q_k^{ij}$ is given by \eqref{Qk}-\eqref{Q0} in Proposition \ref{prp:Gaussian-dtSB} with $k$ replaced by $k \Delta T$ and $N$ replaced by $N \Delta T$. 
The corresponding transition kernel is $p_{k+1|k}^{ij}(x_{k+1} |x_k) = \N(\mu_{k+1|k}^{ij}, \Sigma_{k+1|k}^{ij})$, where 
\begin{align}
    & \mu_{k+1|k}^{ij} = \big(I - \epsilon \Delta T (Q_k^{ij})^{-1}\big) x_{k}^{ij} 
    \label{mu-ij:k+1-k}, \\
    & \Sigma_{k+1|k}^{ij} = \epsilon \Delta T \big(I - \epsilon \Delta T (Q_k^{ij})^{-1}\big).
    \label{sigma-ij:k+1-k}
\end{align}

As $\Delta T \to 0$, the above process converges to the continuous-time optimal process~\cite{chen2018III} for the $ij$-th component, which satisfies
\begin{equation*}
    \d x^{ij}_t = - \epsilon (Q^{ij}_t)^{-1} x^{ij}_t \, \d t + \sqrt{\epsilon} \, \d w^{ij}, 
    \quad t \in [0, T],
\end{equation*}
where $w^{ij}$ is a Wiener process, $Q^{ij}(t)$ satisfies 
\begin{align*}
    Q^{ij}_t &= Q_{0}^{ij} - t \epsilon I,
\end{align*}
and $Q^{ij}_0$ is given by \eqref{Q0} in Proposition \ref{prp:Gaussian-dtSB} with $N$ replaced by $T$.
Thus, the discrete-time optimal SB cost $J_{ij}^{\mathrm{SB}}$ for each $ij$-th component converges to its continuous-time counterpart. 
It follows that the mixing weights $\lambda_{ij}$ also converge to their continuous-time counterparts. 
It follows from \eqref{mflow} that the marginal distribution of the mixed process also converges to its continuous-time counterpart in~\cite{rapakoulias2025go}. 

Given the above, we now study the convergence of the generator of the discrete-time process \eqref{dt_process} and show that it converges to the process defined in~\cite{rapakoulias2025go}. 
Recall from~\cite{rapakoulias2025go} that the continuous-time mixing process is 
\begin{equation} \label{cont-mix}
    \d x_t = \textstyle\sum_{i,j} {u_{t|ij}(x_t) \frac{ \rho_{t|ij}(x_t)\lambda_{ij}}{\sum_{i,j} \rho_{t|ij}(x_t) \lambda_{ij}}} \, \d t + \sqrt{\epsilon} \, \d w_t,
\end{equation}
where $u_{t|ij}(x_t) = - \epsilon (Q^{ij}_t)^{-1} x_t$ and $\rho_{t|ij}(x_t)$ are, respectively, the optimal control and the resulting marginal density of the $ij$-th flow component. 
For notational simplicity, let $\eta_{ij} = \frac{ \rho_{t|ij}(x)\lambda_{ij}}{\sum_{i,j} \rho_{t|ij}(x) \lambda_{ij}}$. 
To this end, consider a sufficiently regular test function $\phi$. 
It is straightforward to see that the infinitesimal generator of \eqref{cont-mix} for the test function $\phi$ is 
\begin{equation*}
    \mathcal{L}(\phi)= - \nabla \phi(x_t)\t \Big(\sum_{ij} \epsilon \eta_{ij} \, (Q_{t}^{ij})^{-1} \Big) x_t 
    + \frac{\epsilon }{2} \tr \big(\hspace{-0.5mm} \nabla^2 \phi(x_t) \big). 
\end{equation*}
Now, compute the limit of the discrete-time generator as $\Delta T \rightarrow 0$. 
We can write the second-order Taylor expansion of $\phi(x_{k+1})$ as
\begin{multline*}
    \phi(x_{k+1}) \approx \phi(x_k) + \nabla \phi(x_k)\t (x_{k+1} - x_k) 
    \\
    + \frac{1}{2} (x_{k+1} - x_k)\t \nabla^2 \phi(x_k) (x_{k+1} - x_k).
\end{multline*}
Furthermore, 
\begin{align*}
    &\lim_{\Delta T\rightarrow0} \frac{ \E [\phi(x_{k+1}) |x_k]- \phi(x_k)}{\Delta T} \\
    & = \lim_{\Delta T\rightarrow0} \frac{ \int \phi(x_{k+1}) p_{k+1|k} \, \d x_{k+1} - \phi(x_k)}{\Delta T} \\
    & = \lim_{\Delta T\rightarrow0} \int \frac{\phi(x_{k+1}) - \phi(x_k)}{\Delta T} p_{k+1|k} \, \d x_{k+1} \\
    & = \lim_{\Delta T\rightarrow0} \int \frac{\phi(x_{k+1}) - \phi(x_k)}{x_{k+1} - x_k} \frac{x_{k+1} - x_k}{\Delta T} p_{k+1|k} \, \d x_{k+1} \\
    & = \lim_{\Delta T\rightarrow0} \int \bigg(\nabla \phi(x_k) + 
    \frac{1}{2} \nabla^2 \phi(x_k) (x_{k+1} - x_k)\bigg)\t
    \nonumber \\
    & \hspace{40mm}
    \times 
    \frac{x_{k+1} - x_k}{\Delta T} p_{k+1|k} \, \d x_{k+1} \\
    & = \nabla \phi(x_t)\t \lim_{\Delta T\rightarrow0} \E\bigg[\frac{x_{k+1} - x_k}{\Delta T} \bigg| x_k\bigg]
    \nonumber \\
    & 
    + \frac{1}{2} \tr \bigg(\hspace{-0.5mm} \nabla^2 \phi(x_t) \hspace{-0.5mm} \lim_{\Delta T\rightarrow0} \hspace{-0.5mm} \E\bigg[\frac{(x_{k+1} - x_k) (x_{k+1} - x_k)\t}{\Delta T} \bigg| x_k\bigg]\bigg).
\end{align*}
In light of the mixing transition kernel \eqref{mtrans}, we define $\gamma_{ij} = \frac{\lambda_{ij} p^{ij}_k(x_k)}{p_k(x_k)}$ for notational simplicity. 
Since the mixing transition kernel $p_{k+1|k}(x_{k+1} |x_k)$ follows a GMM distribution, where the $ij$-th component has density $p_{k+1|k}^{ij}(x_{k+1} |x_k)$ and mixing weight $\gamma_{ij}$, it follows that
\begin{align*}
    \E\bigg[\frac{x_{k+1} - x_k}{\Delta T} \bigg| x_k\bigg]
    & = \frac{1}{\Delta T} \bigg(
    \sum_{ij} \gamma_{ij} \mu_{k+1|k}^{ij} - x_k \bigg)
    \\
    & = -  \Big(\sum_{ij} \epsilon \, \gamma_{ij}\,  (Q_{k}^{ij})^{-1}\Big) x_{k}.
\end{align*}
As $\Delta T\rightarrow0$, we have $k \Delta T \rightarrow t$, $x_{k}^{ij} \rightarrow x^{ij}_t$, and $Q_{k}^{ij} \rightarrow Q_{t}^{ij}$. 
As $\Delta T \rightarrow 0$, we also have $p^{ij}_k(x_k) \to \rho_{t|ij}(x_t)$ and the discrete-time mixing weights $\lambda_{ij}$ converge to their continuous-time counterparts, 
thus $\gamma_{ij} \rightarrow \eta_{ij}$. 
Hence, the first term of the generator becomes 
\begin{multline*}
    \nabla \phi(x_k)\t \lim_{\Delta T\rightarrow0} \E\bigg[\frac{x_{k+1} - x_k}{\Delta T} \bigg| x_k\bigg]
    \\
    = - \nabla \phi(x_t)\t \Big(\sum_{ij}  \epsilon \, \eta_{ij} (Q_{t}^{ij})^{-1}\Big) x_t.
\end{multline*}
Moreover, 
\begin{align}
    &\E\big[(x_{k+1} - x_k) (x_{k+1} - x_k)\t \big| x_k\big]
    \nonumber \\
    & = 
    \sum_{ij} \gamma_{ij} \Big(\Sigma_{k+1|k}^{ij} + \big(\mu_{k+1|k}^{ij} - x_k\big) \big(\mu_{k+1|k}^{ij} - x_k\big)\t\Big) 
    \nonumber \\
    &\hspace{5mm}
    - \Big(\sum_{ij} \gamma_{ij} \big(\mu_{k+1|k}^{ij} - x_k\big)\Big) \Big(\sum_{ij} \gamma_{ij} \big(\mu_{k+1|k}^{ij} - x_k\big)\Big)\t
    \nonumber \\
    & = 
    \sum_{ij} \gamma_{ij} \Big(\epsilon \Delta T \big(I - \epsilon \Delta T (Q_k^{ij})^{-1}\big) \\
    & \hspace{5mm} + \sum_{ij} \epsilon^2 (\Delta T)^2 (Q_k^{ij})^{-1} x_k x_k\t (Q_k^{ij})^{-1}\Big)
    \nonumber \\
    &\hspace{5mm}
    - \epsilon^2 (\Delta T)^2 \Big(\sum_{ij} \gamma_{ij} (Q_k^{ij})^{-1} x_k\Big) \Big(\sum_{ij} \gamma_{ij} (Q_k^{ij})^{-1} x_k\Big)\t.
    \label{2nd-E}
\end{align}
After dividing \eqref{2nd-E} by $\Delta T$ and taking the limit as $\Delta T \rightarrow 0$, all terms containing $(\Delta T)^2$ vanish in \eqref{2nd-E}. 
It follows that 
\begin{equation*}
    \lim_{\Delta T\rightarrow0} \frac{1}{\Delta T}\E\big[(x_{k+1} - x_k) (x_{k+1} - x_k)\t \big| x_k\big]
    =
    \sum_{ij} \eta_{ij} \, \epsilon
    =
    \epsilon. 
\end{equation*}
Thus, the second term of the generator becomes 
\begin{align*}
    &\frac{1}{2} \tr \bigg(\hspace{-0.5mm} \nabla^2 \phi(x_k) \hspace{-0.5mm} \lim_{\Delta T\rightarrow0} \hspace{-0.5mm} \E\bigg[\frac{(x_{k+1} - x_k) (x_{k+1} - x_k)\t}{\Delta T} \bigg| x_k\bigg]\bigg)
    \\
    & = \frac{\epsilon }{2} \tr \big(\hspace{-0.5mm} \nabla^2 \phi(x_t) \big). 
\end{align*}
Therefore, as $\Delta T \rightarrow 0$, the generator of the discrete-time mixing process $p$ converges to the infinitesimal generator of its continuous-time counterpart process in~\cite{rapakoulias2025go}.

\section{Numerical Examples}
Consider two GMMs
\begin{equation} \label{GMMs}
\rho_0 = \sum_{i=1}^{N_1} \alpha^i \N(\mu_0^i, \Sigma_0^i), \quad \rho_N = \sum_{j=1}^{N_2} \beta^j \N(\mu_N^j, \Sigma_N^j),
\end{equation}
with parameters $\{\alpha^i, \mu_0^i, \Sigma_0^i\}_{i=1}^{N_1}, \{\beta^j, \mu_N^j, \Sigma_N^j\}_{j=1}^{N_2}$ representing weights, means and covariances of each component, respectively, and a prior process represented by the LTV system \eqref{dtDS:LTI}. 

\paragraph{Example 1}: For the first example, we let $N=10$, $\epsilon=0.1/N$, $A=B=I, D=\sqrt{\epsilon}$, and solve \eqref{dtSB} and \eqref{dtDS} for $\rho_0, \rho_N$ of the form \eqref{GMMs} with $N_1 = 1, \alpha^i=1, \mu_0^i = [0;  0], N_2 = 8, \beta^{j}=\frac{1}{N_2}, \mu_N^i=5[\cos(\theta_i); \sin(\theta_i)], \, \theta = k \frac{\pi}{8}$ and $\Sigma_0^i = \Sigma_1^j= 0.1 I_2$. We illustrate the result in Figure \ref{fig1}.

\begin{figure}[!ht]
    \centering
    \includegraphics[width=0.9 \linewidth]{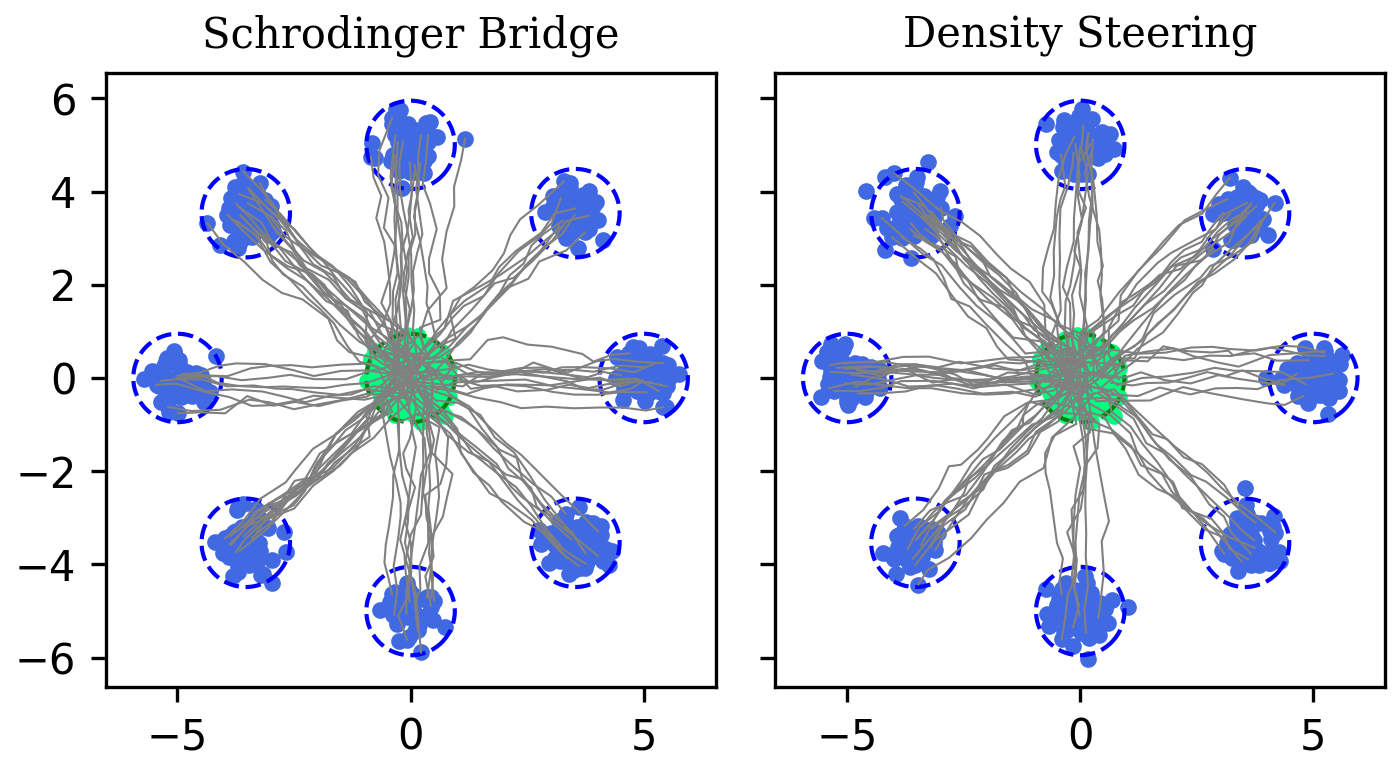}
    \caption{Schrodinger Bridge (left) vs Density steering (right) solutions from initial (green) Gaussian to final (blue) GMM.}
    \label{fig1}
\end{figure}

\vspace{2mm}

\paragraph{Example 2}: For the second example, use $A,B$ for a double integrator~\cite{rapakoulias2023discrete} with a horizon of $N=20$ steps and step size $\Delta T= 1/N$, and solve the discrete-time density steering problem \eqref{dtDS} for $\rho_0, \rho_N$ of the form \eqref{GMMs} with $N_1 = 2, \alpha^i=\frac{1}{N_2}, \mu_0^i = [-5; \{-2,2\}; 20;  0], N_2 = 3, \beta^{j}=\frac{1}{N_2}, \mu_N^i=[5, \{-3, 0, 3\}; 0; 0]$ and $\Sigma_0^i = 0.5 I_4, \Sigma_1^j= 0.2 I_2$. 
We illustrate the result in Figure \ref{fig2}.
\begin{figure}[!ht]
    \centering
    \includegraphics[width=0.8 \linewidth]{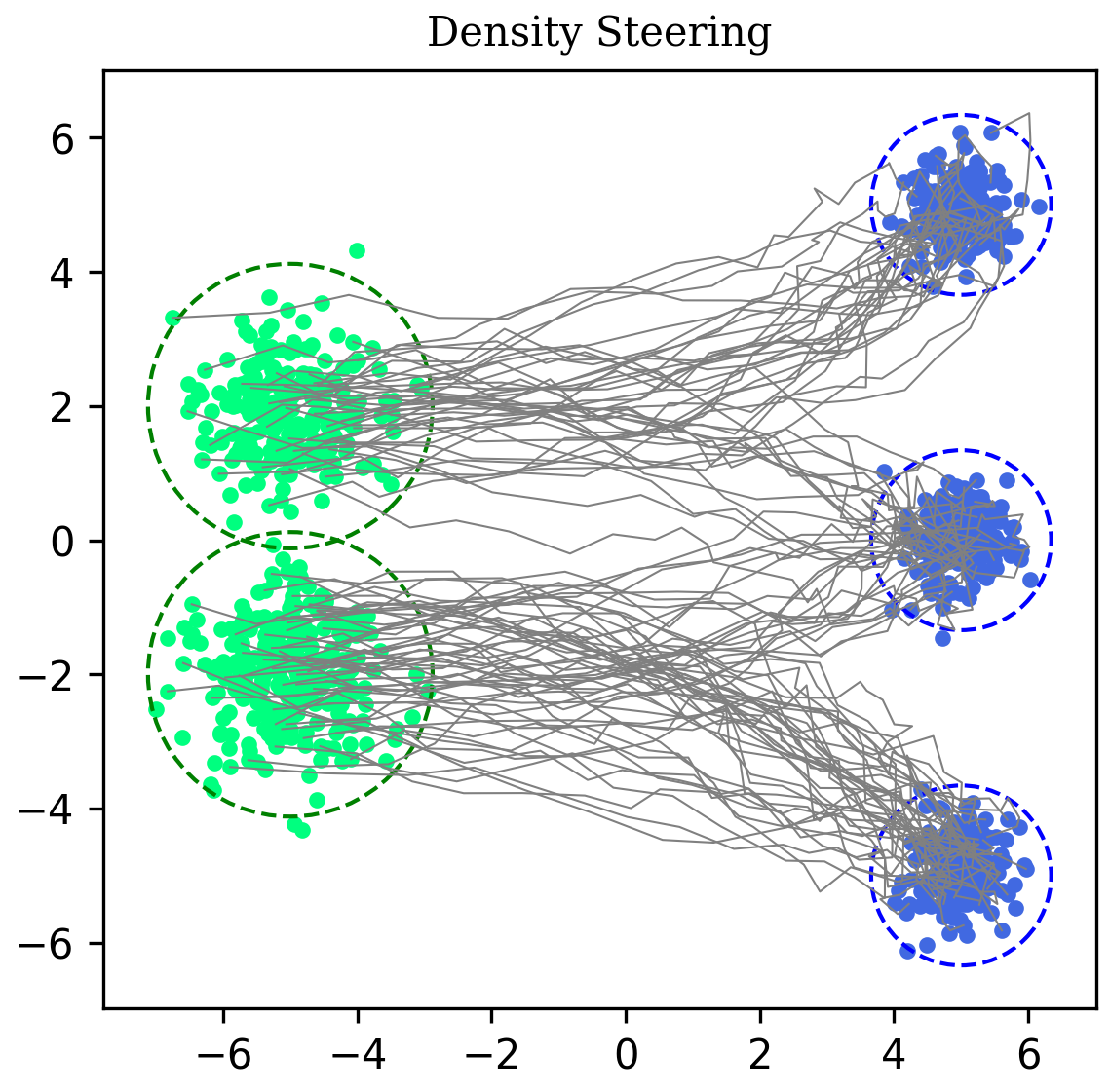}
    \caption{Density steering double integrator dynamics from initial (green) GMM to final (blue) GMM.}
    \label{fig2}
\end{figure}

\section{Conclusions}

To conclude, in this paper, we propose a randomized feedback Markovian policy for approximating the solution of discrete-time Schr\"odinger bridge and density steering problems.
Our method works by first constructing a feasible set of policies, and then approximating the optimal policy within the set by optimizing a tractable upper bound.
We draw connections to existing literature and show that the proposed approach yields a less conservative approximation of SB problems, while achieving the same cost for density control problems, with the additional benefit of being a state-feedback law rather than a controller with memory.
We also study the continuous time limit of the proposed approach, and show that it agrees with a recently proposed continuous-time Schr\"odinger bridge approximation for GMM boundary distributions.
Finally, we illustrate our approach in two examples.

\bibliographystyle{ieeetr}
\bibliography{refs}

@inproceedings{liu2022deep,
 author = {Liu, Guan-Horng and Chen, Tianrong and So, Oswin and Theodorou, Evangelos},
 booktitle = {Advances in Neural Information Processing Systems},
 address={Louisiana, LA},
 pages = {9374--9388},
 publisher = {Curran Associates, Inc.},
 title = {Deep Generalized {Schr\"odinger} Bridge},
 volume = {35},
 year = {2022}
}

@inproceedings{chen2022likelihood,
  title={Likelihood Training of {Schr{\"o}dinger} Bridge using Forward-Backward {SDEs} Theory},
  author={Chen, Tianrong and Liu, Guan-Horng and Theodorou, Evangelos A},
  booktitle={International Conference on Learning Representations},
  year={2022}, 
  month={April},
  address={held virtually}
}

@article{liu2022optimal,
  title={Optimal covariance steering for discrete-time linear stochastic systems},
  author={Liu, Fengjiao and Rapakoulias, George and Tsiotras, Panagiotis},
  journal={Transactions on Automatic Control},
  pages={2289-2304},
  volume={70},
  number={4},
  year={2025}
}

@INPROCEEDINGS{rapakoulias2023discrete,
  author={Rapakoulias, George and Tsiotras, Panagiotis},
  booktitle={62nd Conference on Decision and Control}, 
  title={Discrete-Time Optimal Covariance Steering via Semidefinite Programming}, 
  year={2023},
  volume={},
  address={Singapore},
  pages={1802-1807}}

@article{chen2015optimal,
  title={Optimal Steering of a Linear Stochastic System to a Final Probability Distribution, Part {I}},
  author={Chen, Yongxin and Georgiou, Tryphon T and Pavon, Michele},
  journal={Transactions on Automatic Control},
  volume={61},
  number={5},
  pages={1158--1169},
  year={2015},
  publisher={IEEE}
}

@article{chen2018III,
  title={Optimal Steering of a Linear Stochastic System to a Final Probability Distribution, Part {III}},
  author={Chen, Yongxin and Georgiou, Tryphon T. and Pavon, Michele},
  journal={Transactions on Automatic Control},
  volume={63},
  number={9},
  pages={3112--3118},
  year={2018},
  publisher={IEEE}
}

@article{balci2022exact,
  title={Exact {SDP} formulation for discrete-time covariance steering with {Wasserstein} terminal cost},
  author={Balci, Isin M and Bakolas, Efstathios},
  journal={arXiv preprint arXiv:2205.10740},
  year={2022}
}

@article{bakolas2018finite,
  title={Finite-horizon covariance control for discrete-time stochastic linear systems subject to input constraints},
  author={Bakolas, Efstathios},
  journal={Automatica},
  volume={91},
  pages={61--68},
  year={2018},
  publisher={Elsevier}
}

@inproceedings{balci2023density,
  title={Density Steering of {G}aussian Mixture Models for Discrete-Time Linear Systems},
  author={Balci, Isin and Bakolas, Efstathios},
  booktitle={American Control Conference (ACC)},
  address={Toronto, ON, Canada},
  pages={3935-3940},
  year={2024}
}

@inproceedings{pilipovsky2024computationally,
  title={Computationally efficient chance constrained covariance control with output feedback},
  author={Pilipovsky, Joshua and Tsiotras, Panagiotis},
  booktitle={IEEE 63rd Conference on Decision and Control (CDC)},
  address={Milan, Italy},
  pages={677--682},
  year={2024},
  organization={IEEE}
}

@inproceedings{saravanos2023distributed,
  title={Distributed Hierarchical Distribution Control for Very-Large-Scale Clustered Multi-Agent Systems},
  author={Saravanos, Augustinos D and Li, Yihui and Theodorou, Evangelos A},
  booktitle={Robotics: Science and Systems XIX},
  address={Daegu, Republic of Korea},
  month={July},
  year={2023}
}

@article{chen2021controlling,
  title={Controlling uncertainty},
  author={Chen, Yongxin and Georgiou, Tryphon T and Pavon, Michele},
  journal={Control Systems Magazine},
  volume={41},
  number={4},
  pages={82--94},
  year={2021},
  publisher={IEEE}
}

@article{chen2021stochastic,
  title={Stochastic control liaisons: 
      {Richard Sinkhorn} meets {Gaspard Monge} on a {Schr\"odinger} bridge},
  author={Chen, Yongxin and Georgiou, Tryphon T and Pavon, Michele},
  journal={SIAM Review},
  volume={63},
  number={2},
  pages={249--313},
  year={2021},
  publisher={SIAM}
}

@inproceedings{shi2023diffusion,
  title={Diffusion {Schr{\"o}dinger} bridge matching},
  author={Shi, Yuyang and De Bortoli, Valentin and Campbell, Andrew and Doucet, Arnaud},
  booktitle={Advances in Neural Information Processing Systems},
  volume={36},
  pages = {62183--62223},
  publisher = {Curran Associates, Inc.},
  year={2023}
}

@article{leonard2014survey,
  title={A survey of the {Schr{\"o}dinger} problem and some of its connections with optimal transport},
  author={L{\'e}onard, Christian},
  journal={Discrete \& Continuous Dynamical Systems-A},
  volume={34},
  number={4},
  pages={1533--1574},
  year={2014},
  publisher={American Institute of Mathematical Sciences (AIMS)}
}

@article{terpin2024dynamic,
author = {Terpin, Antonio and Lanzetti, Nicolas and D\"{o}rfler, Florian},
title = {Dynamic Programming in Probability Spaces via Optimal Transport},
journal = {SIAM Journal on Control and Optimization},
volume = {62},
number = {2},
pages = {1183-1206},
year = {2024}
}

@article{chen2018steering,
  title={Steering the distribution of agents in mean-field games system},
  author={Chen, Yongxin and Georgiou, Tryphon T and Pavon, Michele},
  journal={Journal of Optimization Theory and Applications},
  volume={179},
  pages={332--357},
  year={2018},
  publisher={Springer}
}

@article{beghi1996relative,
  title={On the relative entropy of discrete-time {M}arkov processes with given end-point densities},
  author={Beghi, Alessandro},
  journal={IEEE Transactions on Information Theory},
  volume={42},
  number={5},
  pages={1529--1535},
  year={1996},
  publisher={IEEE}
}

@inproceedings{gushchin2024adversarial,
  title={Adversarial {S}chr{\"o}dinger bridge matching},
  author={Gushchin, Nikita and Selikhanovych, Daniil and Kholkin, Sergei and Burnaev, Evgeny and Korotin, Aleksandr},
  booktitle={Advances in Neural Information Processing Systems},
  volume={37},
  pages={89612--89651},
  publisher = {Curran Associates, Inc.},
  year={2024}
}

@inproceedings{kumagai2024chance,
  author={Kumagai, Naoya and Oguri, Kenshiro},
  booktitle={IEEE 63rd Conference on Decision and Control (CDC)}, 
  title={Chance-Constrained {G}aussian Mixture Steering to a Terminal {G}aussian Distribution}, 
  year={2024},
  volume={},
  number={},
  pages={2207-2212},
  address={Milan, Italy},
  doi={10.1109/CDC56724.2024.10886105}}

@article{li2026distributionally,
  author={Li, Yutang and Li, Songzhou and Zhou, Di and He, Zhen},
  journal={IEEE Control Systems Letters}, 
  title={Distributionally Robust {GMM} Steering Under {Wasserstein} Ambiguity Sets}, 
  year={2026},
  volume={10},
  number={},
  pages={13-18},
  doi={10.1109/LCSYS.2026.3657859}}

@article{rapakoulias2025steering,
  author={Rapakoulias, George and Reza Pedram, Ali and Tsiotras, Panagiotis},
  journal={IEEE Control Systems Letters}, 
  title={Steering Large Agent Populations Using Mean-Field {S}chr\"{o}dinger Bridges With Gaussian Mixture Models}, 
  year={2025},
  volume={9},
  number={},
  pages={1760-1765},
  doi={10.1109/LCSYS.2025.3581859}}

@article{balci2025constrained,
  title={Constrained multi-modal density control of linear systems via covariance steering theory},
  author={Balci, Isin M and Bakolas, Efstathios},
  journal={arXiv preprint arXiv:2501.02866},
  year={2025}
}

@inproceedings{rapakoulias2025go,
  title={Go With the Flow: Fast Diffusion for {G}aussian Mixture Models},
  author={Rapakoulias, George and Pedram, Ali Reza and Liu, Fengjiao and Zhu, Lingjiong and Tsiotras, Panagiotis},
  booktitle={The Thirty-ninth Annual Conference on Neural Information Processing Systems},
  year={2025}
}

@inproceedings{kumagai2025hands,
  title={Hands-Off Covariance Steering: Inducing Feedback Sparsity via Iteratively Reweighted $\ell_{1, p}$ Regularization},
  author={Kumagai, Naoya and Oguri, Kenshiro},
  booktitle={IEEE 64th Conference on Decision and Control (CDC)},
  date={December},
  address={Rio De Janeiro, Brazil},
  pages={3560--3565},
  year={2025},
  organization={IEEE}
}

@inproceedings{mei2024flow,
  title = 	 {Flow matching for stochastic linear control systems},
  author =       {Mei, Yuhang and Al-Jarrah, Mohammad and Taghvaei, Amirhossein and Chen, Yongxin},
  booktitle = 	 {Proceedings of the 7th Annual Learning for Dynamics {\&} Control Conference},
  pages = 	 {484--496},
  year = 	 {2025},
  volume = 	 {283},
  series = 	 {Proceedings of Machine Learning Research},
  month = 	 {04--06 Jun},
  publisher =    {PMLR}}

@article{levy1994discrete,
    title = {Discrete-time {Gauss}-{Markov} {Processes} with {Fixed} {Reciprocal} {Dynamics}},
    language = {en},
    journal = {Markov Processes},
    author = {Levy, Bernard C and Beghi, Alessandro},
    year = {1994},
}

@book{cover2012elements,
  title={Elements of Information Theory},
  author={Cover, T.M. and Thomas, J.A.},
  isbn={9781118585771},
  lccn={2005047799},
  url={https://books.google.com/books?id=VWq5GG6ycxMC},
  year={2012},
  publisher={Wiley}
}

@article{ito2023maximum,
  title={Maximum entropy optimal density control of discrete-time linear systems and {S}chr{\"o}dinger bridges},
  author={Ito, Kaito and Kashima, Kenji},
  journal={Transactions on Automatic Control},
  volume={69},
  number={3},
  pages={1536--1551},
  year={2023},
  publisher={IEEE}
}

@inproceedings{mei2025time,
  title={A time-reversal control synthesis for steering the state of stochastic systems},
  author={Mei, Yuhang and Taghvaei, Amirhossein and Pakniyat, Ali},
  booktitle={IEEE 64th Conference on Decision and Control (CDC)},
  pages={1265--1272},
  address={Rio De Janeiro, Brazil},
  year={2025},
  organization={IEEE}
}

@inproceedings{lambert2025lqr,
  title={The LQR-Schr{\"o}dinger Bridge},
  author={Lambert, Marc},
  booktitle={2025 IEEE 64th Conference on Decision and Control (CDC)},
  address={Rio De Janeiro, Brazil},
  pages={3149--3156},
  year={2025},
  organization={IEEE}
}
\end{document}